\documentclass[12pt]{article}
\usepackage{amssymb,amsmath,amsthm,graphicx,ulem}
\usepackage[linktoc=page]{hyperref}

\usepackage{graphicx,subfigure}
\usepackage{epsfig}
\usepackage{amsmath}
\usepackage{amsfonts}
\usepackage{amssymb}
\usepackage[usenames]{color}
\usepackage[letterpaper,left=2.cm,right=2.cm,top=2.5cm,bottom=2.5cm]{geometry}
\usepackage[T1]{fontenc}
\newtheorem{thm}{Theorem}[section]

\newtheorem{cor}[thm]{Corollary} 

\theoremstyle{remark}

\newcommand{\beq}{\begin{equation}}
\newcommand{\eeq}{\end{equation}}
\newcommand{\be}{\begin{equation}}
\newcommand{\ee}{\end{equation}}
\newcommand{\beqa}{\begin{eqnarray}}
\newcommand{\eeqa}{\end{eqnarray}}
\newcommand{\beqar}{\begin{eqnarray*}}
\newcommand{\eeqar}{\end{eqnarray*}}
\newcommand{\bea}{\begin{eqnarray}}
\newcommand{\eea}{\end{eqnarray}}

\newcommand{\hilight}[1]{\colorbox{yellow}{#1}}

\numberwithin{equation}{section}

\numberwithin{equation}{section}

\begin{document}

\allowdisplaybreaks

\normalem

\title{Quantum Extremal Surfaces: Holographic Entanglement Entropy beyond the Classical Regime}

\author{Netta Engelhardt and Aron C. Wall
\\ 
\\ \\
  \textit{Department of Physics}\\
\textit{University of California, Santa Barbara}\\
\textit{ Santa Barbara, CA 93106, USA} \\ 
 \\ 
\small{engeln@physics.ucsb.edu, aroncwall@gmail.com}}

\date{\ }

\maketitle

\begin{abstract}
We propose that holographic entanglement entropy can be calculated at arbitrary orders in the bulk Planck constant using the concept of a ``quantum extremal surface'': a surface which extremizes the generalized entropy, i.e. the sum of area and bulk entanglement entropy. At leading order in bulk quantum corrections, our proposal agrees with the formula of Faulkner, Lewkowycz, and Maldacena, which was derived only at this order; beyond leading order corrections, the two conjectures diverge. Quantum extremal surfaces lie outside the causal domain of influence of the boundary region as well as its complement, and in some spacetimes there are barriers preventing them from entering certain regions.  We comment on the implications for bulk reconstruction.


\end{abstract}

\newpage


\tableofcontents

\baselineskip16pt

\section{Introduction}\label{intro}
The entropy of a surface is proportional to its area.  This is a mysterious truth which appears in several different contexts.

The first context in which this was noticed is black hole thermodynamics.  There, it was observed that black holes in classical general relativity obey various laws of thermodynamics so long as they are assigned an entropy proportional to the area of their event horizon: $S_H = A_H/4G\hbar$.  For example, Hawking's area law \cite{Hawking71} shows that $S_{H}$ is nondecreasing with time, in accordance with the Second Law of thermodynamics.   However, there are also quantum corrections to the area theorem---Hawking radiation implies that black holes are intrinsically quantum objects, and therefore their entropy cannot be entirely described by classical geometry.  In the semiclassical setting, the black hole entropy is given by the so-called ``generalized entropy'' of the event horizon:
\begin{equation}
S_\mathrm{gen}(H) = \frac{\langle A (H) \rangle}{4G\hbar} + S_\mathrm{out} + \mathrm{counterterms}, \label{Sgen}
\end{equation}
where $\left \langle A(H)\right\rangle$ is the expectation value of the area operator on a spatial slice $H$ of the horizon, and $S_\mathrm{out}$ is the von Neumann entropy $-\mathrm{tr}(\rho\,\ln\,\rho)$ for the state $\rho$ of matter fields outside the black hole (e.g. stars or Hawking radiation).  At leading order in $\hbar$, $S_{\text{gen}}(H)=S_{H}$. The generalized entropy obeys the Generalized Second Law (GSL), which states that $S_{\text{gen}}$ is non-decreasing with time \cite{Bekenstein73, Hawking75}. Note that $S_\mathrm{out}$ includes a divergent component due to the vacuum entanglement entropy of short-wavelength modes across the horizon. It is a standard result in quantum field theory that the leading-order divergence is in itself proportional to the area \cite{Srednicki93} and corresponds to a renormalization of Newton's constant \cite{SusskindUglum}. There are also various subleading divergences extensive on the horizon.  These divergences must be absorbed into counterterms, including various subleading quantum corrections to $S_H$ \cite{33}.  The area term is thus merely the dominant ``classical'' contribution to the entropy.

A second context in which the entropy-area relation arises is via  AdS/CFT . AdS/CFT, or gauge/gravity duality, is a correspondence between string theories in asymptotically Anti-de Sitter (AdS) bulk spacetimes and certain conformal gauge theories (CFT's) living on the conformal boundary \cite{Maldacena97}.  In the limit where the CFT is strongly coupled and has a large number $N$ of colors, the bulk theory becomes classical general relativity, coupled to certain matter fields obeying the null energy condition.

In this classical limit, there is strong evidence that the entanglement entropy of a region $R$ in the CFT can be computed from a spacelike, codimension 2 extremal surface $X$ such that $\partial X =\partial R$ and $X$ is homologous to $R$, as proposed by Hubeny, Rangamani, and Takayanagi (HRT) \cite{HRT}.   When the spacetime is static, so that there is a preferred time foliation, the extremal surface is minimal on a constant time slice \cite{RT}. The entanglement entropy of a region $R$ in the CFT is then proportional to the area of the surface in the bulk $S(R) = A(X)/4 G\hbar$.  Given some reasonable assumptions, this formula was recently proven by Lewkowycz and Maldacena (LM) \cite{LewkowyczMaldacena}.\footnote{Out of caution, LM only claim to have proven the static version of the conjecture, which involves minimal area surfaces.  To prove the argument in the non-static case, one would have to analytically continue to complex manifolds.  Although there might be subtleties in the analytic continuation, so far as we can tell, the same argument should also work for extremal surfaces in non-static but analytic spacetimes.}

In complete analogy with black hole thermodynamics, this result is only valid at $\mathcal{O}(\hbar^{-1})$, or equivalently $\mathcal{O}(1/N^{2})$ in the CFT. The first quantum corrections to this formula were computed, at order $\mathcal{O}(N^{0})$ (boundary) or $\mathcal{O}(\hbar^{0})$ (bulk), by Faulkner, Lewkowycz, and Maldacena (FLM); see also \cite{SwingleVanRaamsdonk, Barrella:2013wja}. FLM found that the entropy was given by \cite{FaulknerLewkowyczMaldacena}:
\begin{equation}\label{introFLM}
S_{R} = \frac{\langle A (X) \rangle}{4G\hbar} + S_\mathrm{ent} + \mathrm{counterterms} = S_\mathrm{gen}(X).
\end{equation}
Here $S_\mathrm{ent}$ is the \emph{bulk} entanglement entropy across the surface $X$.  (FLM assumed the geometry was static, but we will assume in what follows that everything carries over to the nonstatic case if $X$ is an extremal surface.) Since FLM restrict their attention to a context in which the total state is pure, $S_\mathrm{out}$ and $S_\mathrm{ent}$ are interchangeable.

The similarity of Eq. \ref{introFLM} with that for black hole entropy \ref{Sgen} is striking.  As in the case of a black hole, we once again find that quantum effects require us to replace the area with a generalized entropy $S_\mathrm{gen}$, now evaluated on an extremal surface $X$ rather than on a slice of an event horizon.\footnote{On a generic manifold obeying the null energy condition, an extremal surface $X$ is \emph{never} a slice of a causal horizon, since a null surface shot out from $X$ has decreasing area by the Raychaudhuri equation, while causal horizons have increasing area \cite{HawkingEllis}.}

In fact we can define $S_\mathrm{gen}$ more generally on a very broad class of surfaces.  Although it is tempting to call $S_\mathrm{gen}$ of the horizon ``\emph{the} black hole entropy'', in fact there are many possible surfaces one could choose on any given spacetime background, and these surfaces may also have a statistical interpretation \cite{Sorkin83, Jacobson95, BianchiMyers12, Balasubramanian:2013lsa, Myers:2014jia, Czech:2014wka}.  The only necessary ingredient is an entangling surface $E$, defined as any spacelike codimension 2 surface which divides a Cauchy surface $\Sigma$ into two pieces.  Let $\mathrm{Ext}(E)$ be that side of $\Sigma$ which is outside $E$, and let $\mathrm{Int}(E)$ be the side which is inside $E$.  We can then define $S_\mathrm{out}(E)$ as the entanglement entropy in the spatial region $\mathrm{Ext}(E)$, and take $S_{\text{gen}}$ as in Eq. \ref{Sgen}, although here $S_{\text{gen}}$ is evaluated on $E$ as opposed to the spatial slice of a horizon.  See Fig. \ref{EntanglingSurface}. By unitarity, any choice of $\Sigma$ passing through $E$ defines the same entropy, so $S_\mathrm{gen}(E)$ does not depend on the choice of $\Sigma$.\footnote{This statement fails in theories with a gravitational anomaly \cite{Wall:2011kb, WallIqbal}.}  

\begin{figure}[ht]
\begin{center}
\includegraphics[width=8cm]{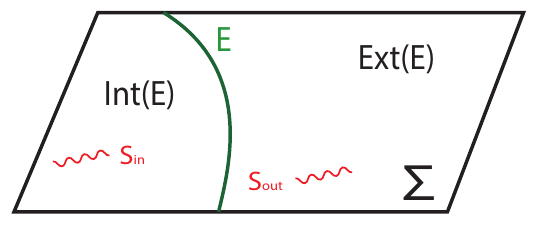} 
\caption{A surface $E$ splits an AdS-Cauchy slice $\Sigma$ into Ext$(E)$ and Int$(E)$. The generalized entropy can be defined with respect to either side, depending on whether we pick $S_{in}$ or $S_{out}$ to calculate $S_{\text{gen}}$. In the case where the state is pure, the choices yield identical results.
}
\label{EntanglingSurface}
\end{center}
\end{figure}

If the state is pure, then $S_\mathrm{out}(E)$ is equal to the entropy $S_\mathrm{in}$ of $\mathrm{Int}(E) $; if the state is mixed we must choose which side to consider.  For purposes of the FLM formula, it seems clear that one should choose the entropy $S$ which lies on the \emph{same side} as the region $R$ whose boundary entropy is of interest.  Choosing the other side would violate locality, since throwing a mixed qubit into the bulk from the complementary boundary region $\bar{R}$ would immediately affect $S(R)$, which is unphysical. Note that in an abuse of notation, we shall use the term entanglement entropy to refer to the entropy of both pure and mixed states.

Having observed that the generalized entropy can be defined for arbitrary surfaces, we now propose a modification to the FLM formula---besides extending to extremal surfaces in non-static spacetimes---namely, that instead of a) extremizing the area and then adding $S_\mathrm{out}$ as FLM did  \footnote{Note that FLM only proposed this prescription at $\mathcal{O}(\hbar^{0})$ in the entropy. However, we will refer to it as the FLM formula at all orders.}, one should instead b) extremize the total generalized entropy $S_\mathrm{gen}$.  We will call such a surface a \emph{quantum extremal surface} since it is a quantum deformation of the usual notion of an extremal surface, in which one extremizes the area.  It will be shown below that prescriptions (a) and (b) are equivalent at the order of the first quantum corrections ($\mathcal{O}(\hbar^{0})$).  Since this was the order of the FLM proof, it does not distinguish which of (a) or (b) is correct at higher order.  At higher, potentially infinite order in $\hbar$, we will argue that (a) is not invariant under boundary unitary transformations and is therefore not the entanglement entropy.  
Furthermore, we will prove several suggestive theorems about (b), which do not hold for (a).  This gives us confidence that the prescription (b) is correct, and that quantum extremal surfaces are connected to key physical properties of quantum spacetimes.

The quantum results we will present all have analogues in classical general relativity, for classical extremal surfaces.  Some of these results are needed for the consistency of the holographic entropy conjecture; others are useful for identifying the regions of spacetimes probed by the extremal surfaces.

In particular, quantum extremal surfaces, like their classical counterparts, are constrained to lie further into the bulk than the causal surface.  In the classical case, it is known that the extremal surfaces $X_R$ do not intersect the causal wedge $W_R = I^-(D_R) \cap I^+(D_R)$ associated with the domain of dependence $D_R$ of a boundary region $R$, on any spacetime obeying the null energy condition \cite{WallMaximin, HubenyRangamani}.  In fact it lies deeper into the bulk in a spacelike direction, so that it cannot even intersect the domain of influence $I_R = I^-(D_R) \cup I^+(D_R)$; the same holds for the extremal surface anchored to the complement of $R$ \cite{WallMaximin}.

We will show that these results continue to hold for quantum extremal surfaces in the perturbative quantum gravity regime, even though the null energy condition is no longer satisfied!  As a consequence, it follows that no bulk signals sent in from $D_R$ can change the spacetime behind the quantum extremal surface, so long as the bulk signals propagate locally.  This is true even if one uses time folds to create signals located outside of the causal wedge, as in the Shenker-Stanford construction \cite{ShenkerStanford}.

On the other hand, there are also bounds limiting the reach of classical extremal surfaces. For example, in \cite{EngelhardtWall}, we showed that spacelike extremal surfaces cannot propagate past any codimension 1 surface with negative extrinsic curvature (assuming that the extremal surfaces can be deformed continuously so as to lie outside). It turns out that a perturbatively quantum spacetime features analogous ``barrier surfaces'': any null surface on which the generalized entropy is non-increasing is an obstacle for quantum extremal surfaces.

These results place limits on certain methods of bulk reconstruction in AdS/CFT.  We will show that the bulk region behind a quantum extremal surface $\mathcal{X}_R$ cannot be accessed from $R$ by any of the types of boundary observables we consider, namely entanglement entropy and local causal signals.

The paper is structured as follows: we introduce necessary terminology in section \ref{Definitions} and technical assumptions on non-classical geometry as well as a key theorem from \cite{Wall12} in section \ref{Premises}. Section \ref{EE} details our prescription for computing holographic entanglement entropy, including some consistency checks and comparison to the FLM formula. In section \ref{reach}, we prove that quantum extremal surfaces lie deeper than (and spacelike to) the causal surface, and additionally act as an obstacle to causal signals.  We also ``quantize'' the theorem that the apparent hroizon always lies inside the causal horizon.  We comment on the implication for bulk signals coming in from the boundary.  Section \ref{Barriers} contains several theorems qualifying barriers to quantum extremal surfaces, and section \ref{discussion} discusses the implications for bulk reconstruction.

%
%

\section{Preliminaries}

\subsection{Definitions}\label{Definitions}

For any future-infinite timelike (or null) worldline $W^+$ (i.e. an observer), we can define a future causal horizon as the boundary of the past of $W^+$, i.e. $\partial I^{-}(W^+)= H^+$ \cite{JacobsonParentani}.  This definition is broad enough to include not only black holes, but also Rindler and de Sitter like horizons, to which the same laws of horizon thermodynamics also apply.  

The Generalized Second Law (GSL) is the statement that the generalized entropy of a causal horizon is nondecreasing in time.  More precisely, if $\Sigma$ is a Cauchy surface then $H=\Sigma \cap H^+$ is a horizon slice.  In its differential form \cite{Wall:2009wm}, the GSL says that 
\begin{equation}\label{gsl}
\frac{\delta S_\mathrm{gen}(H)}{\delta H^{a}}k^{a}\geq 0
\end{equation}
where $\delta H^{a}(p)$, $p \in H$, is a normal vector field living on the surface $H$ defining a first order variation of $H$ along its normal directions, and $k^{a}$ is a future-pointing null vector parallel to the null generators of $H^+$.  

It is also possible to extend the GSL to the case in which there is a set of future-infinite worldlines, and $H^+ = \partial I^+(\bigcup W^+)$ \cite{Wall12}.  In the context of AdS/CFT, this allows us to apply the GSL to causal horizons of boundary spacetime regions.  One can also invoke the time-reverse of the GSL, which says that $S_\mathrm{gen}$ is decreasing with time for past horizons $H^- = \partial I^+(\bigcup W^-)$.  This follows from the GSL by CPT symmetry \cite{WallANEC}.

For the purpose of defining $\delta H^{a}$ and $S_\mathrm{out}$, it is conventional to define the ``outside'' as the region in which $W^+$ lies.  In cases where the total state is mixed,  $S_\mathrm{out} \ne S_\mathrm{in}$, and in fact the GSL holds regardless of which side one evaluates the entropy on (so long as one is consistent).  That is because strong subadditivity in the form $ S(AC) + S(BC) \ge S(A) + S(B)$ tells us that
\begin{equation}
\frac{\delta S_\mathrm{in}(H)}{\delta H^{a}}k^{a} \ge \frac{\delta S_\mathrm{out}(H)}{\delta H^{a}}k^{a},
\end{equation}
choosing a complete slice of the spacetime $ABC$ such that $C$ is the infinitesimal region on $H^+$ corresponding to a null-futureward variation $\delta H$ along $H^+$, and $A$ and $B$ are inside and outside respectively \cite{WallANEC}.  See Fig. \ref{HorizonSlice}.  We will have occasion to use the GSL for both $S_\mathrm{in}$ and $S_\mathrm{out}$ below; we will write $S_\mathrm{ent}$ when we do not care which.
\begin{figure}
\begin{center}
\includegraphics[width=5cm]{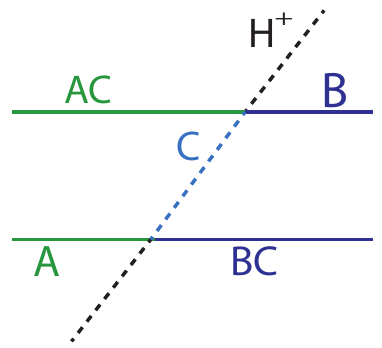} 
\caption{The dashed line represents a future causal horizon $H^{+}$, the solid green lines are regions within the horizon and the solid blue lines are outside it. $C$ is an infinitesimal region on $H^{+}$, along which the entropies $S_\mathrm{out}$ or $S_\mathrm{in}$ are evolved. Strong subadditivity says that $S(AC)-S(A)\geq S(B)-S(BC)$, so that $S_\mathrm{in}$ increases faster than $S_\mathrm{out}$.
}
\label{HorizonSlice}
\end{center}
\end{figure}

Presumably the GSL holds because of the statistics of quantum gravity microstates \cite{Sorkin83, Wall:2009wm, Sorkin:1986mg, Frolov:1993ym, Barvinsky:1994jca}.  The GSL has been proven to hold at least for free bulk fields coupled to semiclassical Einstein gravity \cite{Wall:2009wm}, although there are some mixed results in the case of higher-curvature gravity theories \cite{Sarkar:2013swa, Jacobson:1993xs, Sarkar:2010xp, Liko:2007vi}.  We will assume that the GSL holds in any UV-complete theory of quantum gravity, and will invoke it below even in contexts that go beyond the proof in \cite{Wall:2009wm}.

We may use $S_\mathrm{gen}$ to generalize constructs from classical general relativity to perturbative quantum gravity.  Recall that the classical definition of a (marginally) trapped surface $T$ is a codimension 2 spacelike surface such that the null expansion $\theta$ of future-outward-pointing geodesic congruences from $T$ is negative: $\theta^{+}\leq 0$, with equality for marginally trapped surfaces.  Recall that $\theta = \frac{1}{A}\frac{dA}{d\lambda}$, where $\lambda$ is an affine parameter along the null congruence.\footnote{For purposes of singularity theorems, it is also necessary to assume that the trapped surface is compact, but for purposes of AdS/CFT it is also interesting to consider ``trapped'' surfaces anchored to the boundary.}  Classical trapped surfaces are therefore surfaces whose area decreases along the null congruence: $\frac{dA}{d\lambda}\leq0$.  These conditions can be rephrased as:
\begin{equation} \frac{\delta A}{\delta T^{a}}k^{a}\leq 0\end{equation}
where $\delta T^{a}$ is an infinitesimal variation normal to $T$ and $k^{a}$ is a future-pointing null generator of a geodesic congruence on $T$. Similarly, a classical codimension 2 spacelike extremal surface $X$ is defined to be marginally trapped in both the past- and future- directions, i.e. $\theta^{\pm}=0$, or equivalently:
\begin{equation} \frac{\delta A}{\delta X^{a}}=0\end{equation}
where $\delta X^{a}$ is again an infinitesimal variation normal to $X$. 

While the mathematical notions of extremal and (marginally) trapped surfaces remain well-defined in semiclassical geometries (and presumably even perturbative quantum gravity) they fail to capture key quantum properties of such spacetimes. To define constructs that respect the new, non-classical structure, we must define them in terms of a quantity that is sensitive to quantum fields propagating on the spacetime, but reduces to the area $A$ in the absence of such effects. A natural candidate for this quantity is the generalized entropy, which was previously used in \cite{Wall12} to define a notion of quantum trapped surfaces. We will follow in the same vein and extend the definition in \cite{Wall12} to marginally trapped and extremal surfaces. 

Let $\mathcal{T}$ be a codimension 2 spacelike surface on a Cauchy surface\footnote{Technically AdS space is not globally hyperbolic due to the existence of a timelike boundary at infinity, but its causal properties are still fine assuming that there are boundary conditions at infinity.  We can define a spacelike slice $\Sigma$ in the spacetime $\mathcal{M}$ to be an AdS-Cauchy surface if $\Sigma$ is a Cauchy surface in $\mathcal{M}$ once $\partial \mathcal{M}$ has been conformally compactified. We further define a spacetime to be AdS-hyperbolic, in keeping with \cite{WallMaximin}, if (1) it has no closed causal curves, and (2) for any two points $x$ and $y$ in $\mathcal{M}$, $J^{+}(x)\cap J^{-}(y)$ is compact after conformal compactification of the AdS boundary. These two conditions are equivalent to requiring the existence of an AdS-Cauchy surface. The definitions above allow us to assume a reasonable causal structure without relinquishing relevance to AdS/CFT. } $\Sigma$, and let $\delta \mathcal{T}^{a}$ be as above. If, for any future-directed generator $k^{a}$ of a null congruence on $\mathcal{T}$
\begin{equation} \frac{\delta S_\mathrm{gen}}{\delta \mathcal{T}^{a}}k^{a}<0 \label{quantumtrapped} \end{equation}
then $\mathcal{T}$ is a quantum trapped surface. This definition was first introduced in \cite{Wall12}. A surface $\mathcal{T}$ is therefore quantum trapped if its generalized entropy decreases when it is evolved forwards in time along any future-directed null congruence on $\mathcal{T}$. Since in the classical limit, $S_\mathrm{gen}\propto A$, we recover in this regime the definition of a classical trapped surface. If the inequality sign in Eq. \ref{quantumtrapped} is replaced by an equality, we obtain a quantum marginally trapped surface.
\begin{center}
\begin{figure}[t]
\centering

\subfigure[ ]{
\centering
\includegraphics[width=0.3 \textwidth]{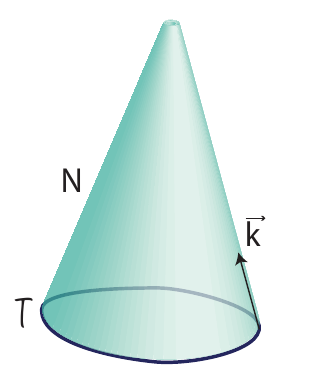}
}
\qquad
\subfigure[]{
\centering \includegraphics[width=0.3\textwidth]{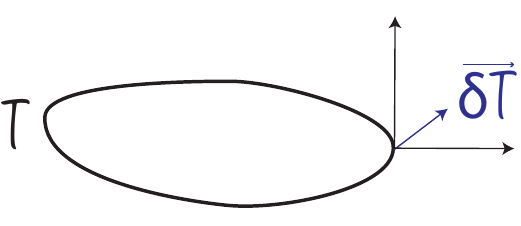}
}
\caption{ (a) {\small $\mathcal{T}$ is a quantum trapped surface with infinitesimal normal $\vec{\delta \mathcal{T}}$, and $N$ is the null surface generated from null congruences shot from $\mathcal{T}$. We could consider variations of $S_{\mathrm{gen}}$ in any direction $\delta \vec{T}$, but we take them only in the direction $\vec{k}$, which is a null generator of $N$. $\mathcal{T}$ is quantum trapped, so we know that the generalized entropy on $\mathcal{T}$ along $\vec{k}$ must be strictly decreasing.} (b) There are two possible directions in which one could deform a codimension 2 surface. $\delta T^{a}$ is defined as any combination of the two.}
\label{Thm3}
\end{figure}
\end{center}

The extension to quantum extremal surfaces follows naturally---just as a classical extremal surface is marginally trapped in both past- and future- directions, a quantum extremal surface $\mathcal{X}$ is quantum marginally trapped in both directions:
\begin{equation} \frac{\delta S_\mathrm{gen}}{\delta \mathcal{X}^{a}}=0\end{equation}
which again reduces to the classical definition in the $\hbar \to 0$ limit.

It is worth pausing at this point to address the utility of the definitions above. Black hole thermodynamics provides a tantalizing hint at a connection between thermodynamics and quantum gravity. The GSL has already been used by one of us to prove the existence of geodesically-incomplete quantum geometries in the presence of quantum trapped surfaces \cite{Wall12}.  A definition of quantum extremal surfaces should prove similarly useful, in particular in the context of AdS/CFT and holographic entanglement entropy in non-classical bulk spacetimes.\footnote{Unlike Ref. \cite{Wall12}, our applications do not involve strong gravity regions near singularities, making our assumption of the GSL more plausible.}

\subsection{Premises about Quantum Geometry}\label{Premises}

There are subtleties in the definition of quantum extremal surfaces that do not arise for extremal surfaces in classical geometry. Such subtleties stem from properties of quantum gravity away from the $\hbar=0$ limit. Below we clarify some of these issues in an attempt to make quantum extremal surfaces well-defined at any order in $\hbar$. Readers who are not concerned with the technical discussion which follows below may wish to skip directly to the statement of Theorem \ref{Wall12Thm1}, which is crucial to the rest of this paper.

We begin by qualifying the regime of validity of the results in this paper.  We work in perturbative quantum gravity via an $\hbar$ \emph{expansion}. A spacetime will be said to admit an $\hbar$ expansion if its metric can be described by an expansion in finite order of $\hbar G/l^{2}$, where $l$ is some length scale of the quantum fields in the theory:
\begin{equation}
g_{ab} = g_{ab}^{(0)} + g_{ab}^{(1/2)} + g_{ab}^{(1)} + g_{ab}^{(3/2)} + \ldots
\end{equation}
where the superscripts represent orders of $\hbar$ after setting $G,\,l = 1$ (the fractional orders arise because quantized gravitons have amplitude $\sqrt{\hbar}$).  If we expand about some specific classical background, then the first term has no fluctuations but the others all do.  The terms in this expansion must therefore be regarded as operators.  (In the semiclassical approximation, one only considers expectation values of the $g_{ab}^{(1)}$ term; if one also ignores graviton fluctuations (as done by e.g. FLM \cite{FaulknerLewkowyczMaldacena}) then this is the first quantum correction to the metric, resulting in an $\mathcal{O}(\hbar^{0})$ contribution to the generalized entropy.)

The terms in the $\hbar$ expansion can be calculated by iterative quantization and backreaction of all fields, including gravitons.  We may obtain an effective field theory at this level by introducing a UV cutoff at energies much smaller than the Planck scale.  We assume without proof that this effective field theory can be consistently defined, and that the GSL holds in it (as stated above).  


The generalized entropy receives radiative corrections due to loop divergences.  These divergences must be absorbed into counterterms, which are of subleading order in $\hbar$ compared to the classical area term.  $S_{\text{ent}}$ is likewise subleading in $\hbar$ since its leading order contribution is at $\hbar^0$.  We can also include $\alpha'$ corrections in a similar expansion to describe a perturbatively stringy spacetime.  This would produce additional subleading higher-curvature corrections to $S_{\text{gen}}$,  which can be calculated for actions which are arbitrary functions of the Riemann tensor by using the Dong entropy formula \cite{Dong}\footnote{The entropy formulae in more restricted cases are given by \cite{Sarkar:2013swa, Fursaev:2013fta, Camps:2013zua, Bhattacharyya:2013gra, Bhattacharyya:2014yga}.}.  Either way, the counterterms are subleading with respect to some parameter.  This allows the counterterms to be consistently neglected in the results that follow.  When proving inequalities in a regime with a small expansion parameter, it is sufficient to prove the result at the first  order at which it is not saturated, since that dominates over all higher terms.  In the context of our proofs, it can be shown \cite{Wall12} that whenever the counterterms would be important, the Bekenstein-Hawking term is always more important, due to being lower order.  The same does not necessarily hold for $S_\mathrm{ent}$, because it is possible to find situations where an order $\hbar$ correction to the area is balanced against the entanglement entropy.

In order to define the notion of a quantum extremal surface, we need to face the unpleasant fact that quantum fluctuations of a spacetime can involve superpositions of different geometries.  As with any gauge theory, this leads to a drastic enhancement of the diffeomorphism group, because each term in the quantum superposition can be separately coordinatized.  For example, in a quantum superposition of a black hole and a neutron star, we could define the ``r'' coordinate independently for each classical metric.  Even if we fix the classical background metric (and we need not!), the issue of quantum superpositions arises at the next order in $\hbar$.  

Often this problem is dealt with by gauge fixing, but fixing the gauge introduces Faddeev-Popov ghosts with negative norm, which is problematic as it introduces negative contributions to the entanglement entropy, violating quantum inequalities.  So it is better to define the quantum extremal surface ${\cal X}$ in a gauge-independent (i.e. covariant) way.  We expect that this can be done in a precise manner, but in this article our main concern is not with quantum superpositions, but rather with the effects of deforming the definition of an extremal surface by adding $S_\mathrm{ent}$.  Thus we merely \emph{outline} a possible approach:

It is instructive to start with the simpler case in which we extremize just the area.  In this case we should promote the area from an expectation value to an operator $\hat{A}$.  (We can likewise promote the counterterms to operators.)  The surface may then be found by demanding that the first order variation of the area vanish as an eigenvalue equation:
\begin{equation}\label{eigen}
\frac{\delta \hat{A}}{\delta X^a} \rho  = 0 = \rho \frac{\delta \hat{A}}{\delta X^a},
\end{equation}
where $\rho$ is a state in a joint Hilbert space
\begin{equation}
\mathcal{H} = \mathcal{H}_\mathrm{bulk}\otimes\mathcal{H}_\mathrm{surfaces},
\end{equation}
where the first factor contains the bulk field theory degrees of freedom, and the second contains the degrees of freedom for (general linear superpositions of) possible locations of the codimension 2 surface.  Since the surface is purely a theoretical construct and not an actual physical quantity, we expect that there will also be an operator $\Omega: \mathcal{H}_\mathrm{surfaces} \to \mathbb{C}$ which we can use to reduce states of $\mathcal{H}$ to states of $\mathcal{H}_\mathrm{bulk}$, the actual physical degrees of freedom.  This ensures that $\rho$ corresponds to the same state of the bulk felds that we started with (the one dual to the chosen CFT state).

In the case of the quantum extremal surface, we cannot simply promote $S_{ent}$ to an operator, since it is not linear in the density matrix $\rho$.  Fortunately, for purposes of defining the extremal surface, we are only interested in the first order variations of $S_{ent}$ with respect to $\delta X^a$.  By definition this first order variation is linear in the perturbations $\delta \rho$, so it corresponds to some linear operator, which could then be inserted into Eq. (\ref{eigen}). 

Having defined the quantum extremal surface $\mathcal{X}$, we can then evaluate $S_\mathrm{gen}(\mathcal{X})$; at this step we only need the expectation values: $\langle A \rangle + S_\mathrm{ent} + \langle \mathrm{counterterms} \rangle$.  This approach requires further investigation, but for now we choose to work under the assumption that it can be made precise in the $\hbar$ expansion.  Hence we will assume that there always exists a gauge-frame in which $\mathcal{X}$, and any other surfaces of interest (e.g. horizons), have a sharp location, and that the classical geometrical relations between such surfaces continue to hold as operator relations in the perturbative regime.

We will now quote a theorem from \cite{Wall12} which we will use throughout this paper to prove a number of results about quantum extremal surfaces. This theorem assumes an $\hbar$ expansion (which is needed to consistently neglect the higher-curvature counterterms, and also to justify a quantum inequality applied to $S_\mathrm{ent}$):

\begin{thm}\label{Wall12Thm1}(from \cite{Wall12}): Let $M$, $N$ be null splitting surfaces (i.e codimension 1 surfaces which divide spacetime into two regions \cite{EngelhardtWall} and have an open exterior) which coincide at a point $p$ and let $\Sigma$ be a spacelike slice that goes through $p$. If (1) $M\cap \text{Ext}(N)=\emptyset$, and (2) $M$, $N$ are smooth at the classical order near $p$, and the spacetime is described by an $\hbar$ expansion there, then there exists a way of evolving $\Sigma$ forward in time in a neighborhood of $p$ so that 
\begin{equation} 
\Delta S_\mathrm{gen}\left (M\right)\geq \Delta S_\mathrm{gen}\left(N\right)
\end{equation}
 with equality only if $M$ and $N$ coincide at a neighborhood.
\end{thm}

\noindent In particular, there exists a normal vector $\delta \Sigma^{a}$ to $\Sigma\cap M$ such that 
\begin{equation} \frac{\delta S_\mathrm{gen}\left(M\right)}{\delta \Sigma^{a}}k^{a} - \frac{\delta S_\mathrm{gen}\left(N\right)}{\delta \Sigma^{a}}k^{a}\geq 0 \label{Diff}\end{equation}
where $k^{a}$ is the null normal to $M$ and $N$ at $p$. Note that 
we use $S_\mathrm{gen}(M)$ to denote the generalized entropy of a spatial slice $M\cap \Sigma$  of $M$. 

\section{A Holographic Entanglement Entropy Proposal} \label{EE}

We will now detail our proposal for computation of entanglement entropy via holography at any order in $\hbar$. We begin by briefly reminding the reader of some recent holographic entropy proposals. 

HRT proposed that the holographic entanglement entropy $S(R)$ is proportional to the area of a classical extremal surface $X_R$ anchored to $\partial R$ and homologous to $R$.  (If there are multiple such extremal surfaces, one must choose the one with the least area.)  When $R$ has a boundary, $X_R$ is noncompact, and $A(X_R)$ is IR divergent.  This bulk IR divergence is dual to the UV divergence of $S(R)$ in the boundary CFT.  In order to test the conjecture, one must compare universal aspects of the divergent entropies (those that do not depend on the choice of regulator).

The FLM formula (\ref{introFLM}) generalizes the classical prescription by considering quantum corrections to the Euclidean gravitational path integral.  FLM showed that the entanglement entropy of a region $R$ on the boundary is given by the generalized entropy of the classical extremal surface, at least at $\mathcal{O}(\hbar^{0})$ and for static situations.

We expect that at higher orders in $\hbar$, the entanglement entropy is given by the generalized entropy of a surface whose location is sensitive to the quantum corrections to the entropy.  Our ideal candidate for a formula for holographic entanglement entropy in a non-classical bulk should therefore (1) be sensitive to quantum effects at any order in $\hbar$, (2) reproduce the classical HRT formula at $\mathcal{O}(\hbar^{-1})$, and (3) reproduce the result of FLM at $\mathcal{O}(\hbar^{0})$. We will show below that the generalized entropy of a quantum extremal surface obeys all three requirements. We propose the following conjecture:\\
\textbf{Conjecture:} The entanglement entropy of a region $R$ in a field theory with a holographic dual is given at any order in $\hbar$ in the holographic dual by the generalized entropy of the quantum extremal surface $\mathcal{X}_{R}$ anchored at $R$ and homologous to $R$:
\begin{equation}
S_{R} = S_\mathrm{gen}\left(\mathcal{X}_{R}\right)\label{Ent}
\end{equation}
If there are multiple such quantum extremal surfaces, then we propose that the one with least $S_\mathrm{gen}$ should be selected, at least when their difference is large: $\Delta S_\mathrm{gen} \gg 1\,\mathrm{bit}$.\footnote{Presumably the surface with least $S_\mathrm{gen}$ dominates in an FLM-like gravitational path integral calculation of $S(R)$ \cite{FaulknerLewkowyczMaldacena}, although if two entropies differ only by an order unity number of bits, then there might be comparable contributions coming from each extremal surface.  The holographic entanglement community is currently puzzled about what the LM argument says when there are multiple extremal surfaces, due to a perplexing argument by Myers \cite{Rob, Fischetti:2014zja, Matt} that found the \emph{average} of the areas.  His argument takes the number of replicas $n \to 1$ first, before taking Newton's constant $G \to 0$.  If one takes the other order of limits, one gets the minimum area surface as expected.}
A noncompact quantum extremal surface also has IR divergences, not only in the area term, but also in $S_\mathrm{out}$ (not to mention the counterterms).  Since the presence of IR divergences is not a new feature of quantum extremal surfaces, we will not worry about it here, but will assume that these divergences are cut off in a suitable manner.

In cases where the total state of the system is mixed, it is necessary to choose one side of ${\cal X}$ to evaluate the entropy on: we may use either $S_\mathrm{out}$ to calculate $S(R)$, or $S_\mathrm{in}$ to calculate $S(\bar{R})$.  In our proposal these correspond to two different bulk surfaces ${\cal X}_\mathrm{out}$ and ${\cal X}_\mathrm{in}$.  This fact that there are now two surfaces is reminiscent of the case of black holes, where there is also a mixed boundary CFT and two extremal surfaces in the bulk, but for a different reason (the homology constraint).

\subsection{Comparison to Faulkner--Lewkowycz--Maldacena Formula}\label{Comparison}

We will now show that our proposal reproduces the FLM result in the regime in which FLM showed it.  Given a region $R$ of a CFT, we can find either the classical extremal surface $X_R$ anchored to $\partial R$, or the quantum extremal surface $\mathcal{X}_{R}$.  Once we include quantum effects, these are no longer the same surface.  The leading order quantum correction to the entropy however, is the same.  In the absence of graviton fluctuations, as in the derivation of FLM:
\begin{equation}\label{agree}
S_\mathrm{gen}(X_{R}) = S_\mathrm{gen}(\mathcal{X}_{R}) + \mathcal{O}(\hbar^{1})
\end{equation}
In a semiclassical spacetime, we expect that the classical extremal surface $X_{R}$ and the quantum extremal surface $\mathcal{X}_{R}$ are a (proper) distance of order $\hbar$ apart.  The difference in $S_\mathrm{ent}(X_{R})$ and $S_\mathrm{ent}(\mathcal{X}_{R})$ is therefore of order $\mathcal{O}(\hbar)$ and can be neglected at this order.  It remains to show that the areas agree at order $\hbar$.  $X_{R}$ and $\mathcal{X}_{R}$ are a distance $\hbar$ apart, but first order variations away from the classical extremal surface $X_R$ do not change the area.  Hence $A(X_{R}) - A(\mathcal{X}_{R})=\mathcal{O}(\hbar^2)$.

The proposals $S_\mathrm{gen}(X_{R})$ and $S_\mathrm{gen}(\mathcal{X}_{R})$ will not agree at higher order.  Calculations at higher order in $\hbar$ (perhaps using the methods of \cite{RosenhausSmolkin14}) could in principle determine which of the two approaches gives the correct formula for the holographic entanglement entropy.   Because it is easier to prove fruitful theorems about the quantum extremal surface $\mathcal{X}_{R}$, we expect that our prescription is more likely to be correct.

For example, we will show in section \ref{reach} that the quantum surface $\mathcal{X}_{R}$ always lies outside the domain of influence of the causal wedge $W_R$.  This is an important consistency relation, since the entropy $S_\mathrm{gen}(\mathcal{X}_{R})$ must be invariant under all unitary transformations of the boundary region $R$.  The classical extremal surface $X_{R}$ does not obey this condition, and we will argue that its entropy can therefore be influenced by unitary operations in $R$ (although there is a sense in which this effect goes beyond the $\hbar$ expansion).

The surface $\mathcal{X}_{R}$ also seems to be more likely to be able to handle radiative corrections.  When the gravitational action is no longer Einstein-Hilbert, it is known that one must extremize, not the area, but a corrected entropy functional \cite{Sarkar:2013swa, Dong, Fursaev:2013fta, Camps:2013zua, Bhattacharyya:2013gra, Bhattacharyya:2014yga, Hung:2011xb, deBoer:2011wk}. It is possible to obtain these corrections from divergences in the entanglement entropy \cite{ 33}.  The choice of UV cutoff determines which contributions are considered ``entanglement'' and which are ``higher curvature corrections''.   An RG flow of the cutoff towards in the infrared reassigns entropy from the former category to the latter.  Consistency under the RG flow therefore requires that these two types of modifications to the holographic entropy must be treated in the same way.

It is not clear whether the FLM result can be extended to include quantized bulk gravitons.  In this case it is possible that the prescriptions for entropy will not agree even at $\mathcal{O}(\hbar^{0})$.\footnote{There is no correction to $S_\mathrm{gen}$ at the intermediate order $\mathcal{O}(\hbar^{-1/2})$.  At this order, the graviton field is linear, so that there is a symmetry $g_{ab}^{(1/2)} \to - g_{ab}^{(1/2)}$.  In any quantum state which preserves this symmetry (e.g. a Hartle-Hawking state), there will be a vanishing expectation value $\langle g_{ab} \rangle = 0 + \mathcal{O}(\hbar^1)$, where the correction is due to nonlinear effects.} Any differences between the two prescriptions would arise due to $X_R$ and ${\cal X}_R$ being at different locations.  Gravitons might produce an order $\hbar^{1/2}$ separation between the two surfaces $X$ and ${\cal X}$.  This does not lead to a discrepancy at $\mathcal{O}(\hbar^{-1/2})$ since $X$ is extremal.  But it might lead to a discrepancy at $\mathcal{O}(\hbar^{0})$. It would be interesting to check this with more explicit calculations.

\section{Quantum Extremal Surfaces lie deeper than Causal Surfaces}\label{reach}

We will now show that the quantum extremal surface cannot intersect the causal wedge.  In fact, it must be spacelike to it, lying deeper in the bulk.

First we review the classical situation.  If the bulk obeys the null energy condition, it is known that the extremal surface $X_R$ for a boundary region $R$ cannot intersect the bulk causal wedge $W_R$, defined as the intersection of the past and future of $R$:
\begin{equation} 
W_R = I^{-}\left (D_{R}\right) \cap I^{+}\left(D_{R}\right). 
\end{equation}
The boundary of this wedge is called the causal surface (see Fig. \ref{CS}):
\begin{equation} 
C_R = \partial I^{-} \left(D_{R} \right) \cap \partial I^{+}\left(D_{R}\right) = \partial_0 W_R.
\end{equation}
Assuming the null energy and generic conditions, $X_R$ lies farther from the boundary than $C_R$, and is spacelike to it \cite{Wall12}.\footnote{Nongenerically, it is possible for $X_R$ and $C_R$ to be null separated or to coincide.}  
In other words, $X_R$ also cannot lie anywhere inside the domain of influence $I_R$, defined as
\begin{equation}
I_R = I^+(D_R) \cup I^-(D_R).
\end{equation}
(The weaker result that $X_R$ does not lie in the wedge $W_R$, was proven in \cite{HubenyRangamani}.)  In the case where $R$ is taken to be an entire asymptotic boundary, this result reduces to the classic theorem that trapped surfaces are enclosed by event horizons \cite{HawkingEllis,Wald}

\begin{figure}[ht]
\begin{center}
\includegraphics[width=5cm]{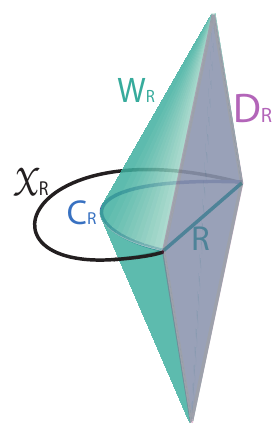} 
\caption{{\small The causal surface $C_R$ is the spacelike boundary of the causal wedge (depicted in green) associated with $R$. The domain of dependence $D_{R}$ of $R$ is depicted in purple. The quantum extremal surface $\mathcal{X}_{R}$ does not intersect the causal wedge.}}
\label{CS}
\end{center}
\end{figure}

This has some important implications for bulk reconstruction.  Using the classical equations of motion, it is possible (at least perturbatively in the bulk coupling constants) to reconstruct the bulk state in $W_R$ from the bulk fields near $D_R$ \cite{HubenyRangamani, 59, Bousso:2012sj}.  More generally, one could measure the bulk fields in $W_R$ by sending in bulk observers who begin and end within $D_R$.

However, since the extremal surface lies outside of $C_R$, it seems that one can in fact obtain some information about the spacetime farther from the boundary than $C_R$ \cite{HubenyRangamani,Czech:2012bh}; plausibly \emph{all} information up to $X_R$ can be reconstructed \cite{Wall12}.

When the bulk becomes quantum, it no longer satisfies the null energy condition.  This implies that $X_R$ could be closer to the boundary than $C_R$, or else timelike separated from it.

This implies that one cannot always reconstruct the bulk up to $X_R$.  For suppose we have a pure state, so that $X_R = X_{\bar R}$, and suppose that ${\bar R}$ can be used to reconstruct anything up to $X_{\bar R}$ on the ${\bar R}$ side.  If $X_R$ were to lie in $I^+(D_R)$, then it would be possible to send a signal from $D_R$ to the reconstructed region, which would imply that observables in $R$ and ${\bar R}$ do not commute, violating microcausality of the CFT.  The same argument in time reverse shows that $X_R$ could not lie in $I^-(D_R)$.  Thus the reconstruction hypothesis is consistent only if the extremal surface is spacelike to $C_R$ and deeper into the bulk.

The solution is to use the \emph{quantum} extremal surface $\mathcal{X}_{R}$ instead.  This surface is deeper than $C_R$ and spacelike to it, as we will now show.  We can no longer use the null energy condition, so we instead prove our result using the GSL (\ref{gsl}), which we assume holds for bulk spacetimes with arbitrary quantum corrections.

\begin{thm} A quantum extremal surface $\mathcal{X}_{R}$ can never intersect $W_R$. 
The surface $\mathcal{X}_{R}$ is moreover generically spacelike separated from the causal surface $C_{R}$, but might be null separated from or coincide with $C_R$ in non-generic spacetimes.
\label{causal}\end{thm}
\begin{proof}
\noindent This proof is a reversed application of the method used in \cite{EngelhardtWall} to prove the existence of classical barriers to extremal surfaces, which is also used below in Theorems  \ref{NegBarrier} and \ref{ZeroBarrier} to prove the existence of quantum barriers. We would like to show that  $C_R$ cannot extend beyond $\mathcal{X}_{R}$, i.e. $C_R\ \cap \ \text{Int}\left(\mathcal{X}_{R}\right)=\emptyset$.  As defined in Section \ref{intro}, Int$\left(\mathcal{X}_{R}\right)$ is that part of the AdS-Cauchy surface which is on the $\bar{R}$ side, but in this proof it is more convenient to use the codimension 0 domain of dependence $\mathbf{Int}(\mathcal{X}_{R}) \equiv D_{\mathrm{Int}\left(\mathcal{X}_{R}\right)}$; similarly for $\mathbf{Ext}(\mathcal{X}_{R})$.

Assume for contradiction that $C_{R}\cap \  \textbf{Int}(\mathcal{X}_{R})\neq \emptyset$.  Consider continuously shrinking $D_R$ to a new boundary spacetime region $\Xi$ (not necessarily a domain of dependence) such that the causal surface $C_\Xi = \partial I^+(\Xi) \cap \partial I^-(\Xi)$ is entirely contained in \textbf{Ext}$(\mathcal{X}_{R})$.  Let $H^+_\Xi = \partial I^-(\Xi)$ be the associated future causal horizon, and $H^-_\Xi = I^+(\Xi)$ the past horizon.  Because the shrinking action is continuous, we can find a choice of $\Xi$ such that $H^\pm_\Xi$ coincides with $\mathcal{X}_{R}$ at some points $\{p\}$ and is tangent to it at those points, and elsewhere lies in \textbf{Ext}$(\mathcal{X}_{R})$.  Without loss of generality, we consider the case in which it coincides with $H^+$; the case of $H^-$ is exactly the same except that we would need to use the time reverse of the GSL below.  This is illustrated in Fig. \ref{Daleth}.

\begin{figure}[ht]
\begin{center}
\includegraphics[width=7cm]{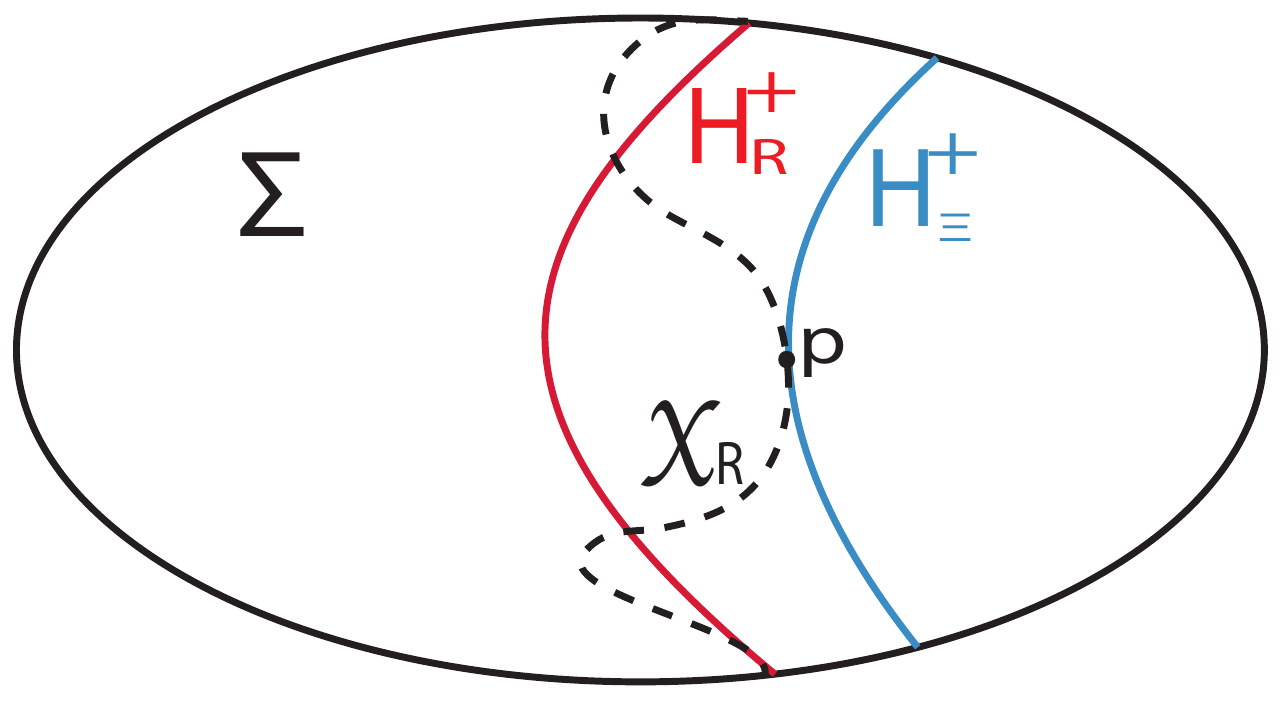} 
\caption{\small{A projection of $H^{+}_{R}$ (depicted in red) and $H^{+}_{\Xi}$ (depicted in blue) onto a slice $\Sigma$ passing through $\mathcal{X}_{R}$ (dashed black).}}
\label{Daleth}
\end{center}
\end{figure}
At any such coincident point $p$, let us functionally differentiate $\mathcal{X}_{R}$ and a slice of $H^+_\Xi$ with respect to their shared normal directions $n^{a}$.  By construction, $H^+_\Xi\subset \text{Ext}(\mathcal{X}_{R})$, so by Eq. \ref{Diff},
\begin{equation}\frac{\delta S_{\text{gen}} (N(\mathcal{X}_{R}))}{\delta n^{a}}k^{a} \geq \frac{\delta S_{\text{gen}} (H^+_\Xi)}{\delta n^{a}}k^{a},  \label{ExtremalCausalComparison}
\end{equation}
where $N(\mathcal{X}_{R})$ is the null surface generated by shooting out light rays from $\mathcal{X}_{R}$ in the future-outwards direction (towards $R$).  By the definition of a quantum extremal surface, the left hand side of Eq. \ref{ExtremalCausalComparison} vanishes, so we find that for future-outward $k^a$,
\begin{equation}
\frac{\delta  S_{\text{gen}} (H^+_\Xi)}{\delta n^{a}}k^{a} \leq 0 \label{CausalNeg},
\end{equation}
with equality only if $\mathcal{X}_{R}$ lies on $H^+_\Xi$ in a neighborhood of $p$.

But the GSL says that
\begin{equation} \frac{\delta  S_{\text{gen}} (H^+_\Xi)}{\delta n^{a}}k^{a} \geq  0, \label{CausalFut} \end{equation}
where the inequality can be saturated only in non-generic situations.  In the generic case, this is in direct contradiction with Eq. \ref{CausalNeg}.  We have thus shown that the quantum extremal surface must lie outside the causal extremal surface, for generic spacetimes.

Suppose now that the spacetime falls under the non-generic case where the inequalities are saturated.  By continuity with the generic conclusion, $\mathcal{X}_{R}$ cannot be outside of either horizon; at worst $\mathcal{X}_{R} \in H^+_\Xi$ or $\mathcal{X}_{R} \in H^-_\Xi$ or both (in which case $\mathcal{X}_{R} = C_R$).  Thus $C_R$ is either spacelike or null separated from $\mathcal{X}_{R}$, and is closer to $R$.  We therefore find that quantum extremal surfaces probe a deeper region of the spacetime than the causal surface does.


\end{proof}

Note that this proof is valid whether we define ${\cal X}_R$ using $S_\mathrm{out}$ or $S_\mathrm{in}$, as long as we use the same side to define the extremal surface and the horizon entropy (since the GSL is valid either way). This, combined with the fact that $\partial R=\partial\bar{R}$, shows that \textit{both} the extremal surface $\mathcal{X}_{R}$ and the complementary extremal surface $\mathcal{X}_{\bar{R}}$ lie deeper in the bulk than the causal surface $C_{R}$. This is a quantum generalization of a theorem in \cite{WallMaximin}, which showed the same for the classical extremal surface $X_{R}$ and $X_{\bar{R}}$.

We can also prove a quantum generalization of the known classical black hole result, that the apparent horizon always lies within the event horizon of a black hole. In order to state this generalization precisely, we define the notion of a quantum apparent horizon\footnote{We thank D. Marolf for pointing this out to us.}.

Let $\Sigma$ be an AdS-Cauchy surface, and let $\mathcal{T}$ be the union of all quantum trapped surfaces on $\Sigma$. Define the \emph{quantum apparent horizon} $\mathcal{H}_{\mathrm{app}}$ to be the boundary of $\mathcal{T}$ (on $\Sigma$).  By Theorem \ref{Wall12Thm1}, $\mathcal{H}_{\text{app}}$ is a quantum marginally trapped surface.  Since in the classical limit, quantum trapped and marginally trapped surfaces reduce to ordinary trapped and marginally trapped surfaces, the quantum apparent horizon reduces to the ordinary (classical) apparent horizon in the limit where $\hbar \rightarrow 0$.

\begin{thm} The quantum apparent horizon always lies inside the horizon. \end{thm}

\begin{proof} The proof follows directly from Theorem \ref{causal}. Since the variation $\delta S_{\text{gen}}/\delta \mathcal{X}^{a}$ only appears when contracted with a null normal, \emph{i.e.} $\left(\delta S_{\text{gen}}/\delta \mathcal{X}^{a}\right) k^{a} $, the proof of Theorem \ref{causal} also applies to marginally quantum trapped surfaces. In particular, a marginally trapped quantum surface $\mathcal{T}$ anchored at $\partial \bar{ R}$ is outside and spacelike to $C_{R}$.\end{proof}

This result, while of interest in its own right, also gives us confidence that quantum extremal, trapped, and marginally trapped surfaces are the correct ways of generalizing the same classical concepts.

\subsection{A limit to causal signaling}

The fact that the quantum extremal surface $\mathcal{X}_{R}$ always lies deeper inside the spacetime than the causal surface $C_R$ has interesting consequences for the resulting geometries.  Because the quantum extremal surface lies outside the causal surface, it limits the farthest extent to which causal sources can propagate in from the boundary.

Suppose we throw some causal signals in from the boundary.  We could do this by acting on the boundary field theory region $D_R$ at some time $t$ with a unitary operator $U(t)$ whose effects on the boundary are purely local.  An example of this, in the limit where the bulk fields are free, would be if $U(t)$ translates a bulk field $\phi$ by some function of the spatial coordinates: $\phi \to \phi + f(x)$.  In the interacting case it is probably necessary to smear out the operator a little bit in the time direction so that it lies within a small time interval $t \pm \epsilon$ (but staying inside $D_R$).  We can then define $U$ by deforming the Hamiltonian by the addition of relevant or marginal operators.  In either case, the resulting pulse will be localized inside the causal wedge $W_R$.

At any finite order in $\hbar$, the resulting signal will propagate locally within the bulk, so it can also reach the past or future of $W_R$.  Thus, if we act with a single unitary operator $U(t)$, this implies that the signal must remain within the causal domain of influence $I_R = I^+(W_R) \cup I^-(W_R)$, and thus does not extend past $C_R$ in a spacelike direction.  So the pulse cannot reach the quantum extremal surface ${\cal X}_R$, since by Theorem \ref{causal}, ${\cal X}_R$ is spacelike outside of $C_R$.  Hence the area of ${\cal X}_R$ cannot be affected, and neither can $S_\mathrm{ent}$ be affected since the operator is unitary.  See Fig. \ref{CausalSignal}.  ${\cal X}_R$, its area, and anything behind it are therefore identical in the perturbed and unperturbed spacetimes.

\begin{figure}[ht]
\begin{center}
\includegraphics[width=5cm]{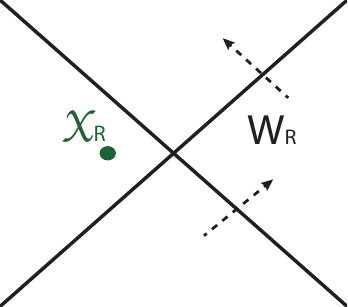} 
\caption{\small{A conformal diagram illustrating the fact that $\mathcal{X}_{R}$ cannot be affected by unitary operators on $R$, since the resulting causal signals remain within the domain of influence $I_R$ (consisting of the top, bottom, and rightmost quadrants)}}
\label{CausalSignal}
\end{center}
\end{figure}

In the cases in which a region $R$ has multiple quantum extremal surfaces ${\mathcal{X}_R}$, one cannot signal past any of them, so among these the tightest bound comes from the one closest to $R$.  This will not necessarily be the one with least entropy, which would be used to calculate the holographic entanglement entropy.  

Similarly, one could act on the CFT boundary with operators which excite $n$-point functions of the fields.  So long as $n$ is order unity in $N$ (otherwise the concept of bulk locality might break down), the fields should still propagate into the bulk causally.

This shows that the region behind ${\cal X}_R$ is invariant under a important class of operators acting on $D_R$: local bulk-causal unitaries.\footnote{One can also analyze sending in signals from the boundary which are local but nonunitary, thus changing the entropy of $R$.  One could do this by coupling $R$ to an auxilliary system by means of a unitary operator.  But beware: these can affect $S_\mathrm{out}$, and therefore also the location of the quantum extremal surface $\mathcal{X}_{R}$!  However, if one defines ${\cal X}_R$ using $S_\mathrm{in}$, the proof goes through.}  

Of course, there is also a classical version of this statement for the classical extremal surface $X_R$, but it is restricted to the case when the 
bulk
obeys the null energy condition, which can be violated by quantum fields.  This could present a consistency problem for FLM: if, in the region outside $C(R)$, the null energy falling across the horizons $H^\pm = \partial^\pm W_R$ is negative, this can cause the extremal surface $X_R$ to lie inside the causal wedge.  The area or $S_\mathrm{out}$ of $X_R$ could then be affected by unitary signals sent in from the boundary, ruling out $X_R$ for use in the holographic entanglement entropy $S(R)$.  However, it should be noted that if the null energy condition violation is of order $\hbar$, then $X_R$ and $\mathcal{X}_R$ will also be separated by a distance of order $\hbar$.  This requires the signal sent in from $D_R$ to be highly boosted relative to the separation, by an amount of order $\hbar^{-1}$.  So this argument should be qualified with the caveat that it may require quantum gravity effects outside the domain of validity of the $\hbar$ expansion.

\subsection{With time folds}

One can also consider the effects of multiple unitary operators using the ``time fold'' construction \cite{MarolfPolchinski}.  Even in this case we can show that causal signals cannot influence the region behind ${\cal X}_R$.

A nice example was provided by Shenker and Stanford \cite{ShenkerStanford}.  They started with an eternal AdS-Schwarzshild black hole, dual to the thermofield double state $|TFD\rangle$.  By acting on this state with a sequence of unitary operators at various times $t_0,\,t_1,\ldots t_n$:
\begin{equation}\label{U's}
U(t_n) \ldots U(t_2) U(t_1) U(t_0) |TFD\rangle,
\end{equation}
they described signals propagating into the bulk from either or both of the two CFT boundaries (CFT$_1$ and CFT$_2$).  In the case where the times are not in sequential order, it is necessary to use the time fold formalism \cite{MarolfPolchinski} to work out the resulting bulk spacetime.  (For ease of calculation, Shenker and Stanford ultimately take the limit where the times go to $t = \pm \infty$, but we will avoid taking this limit here.)

Surprisingly, this can be used to create causally propagating sources that are outside of the domain of influence $I_\mathrm{CFT's} = I_{\mathrm{CFT}_1 + \mathrm{CFT}_2}$.  This is because the unitary operators to the left in Eq. (\ref{U's}) can change the relative location of the earlier signals to the past or future horizon.  The leftmost source always propagates from the boundary, but the others may emerge from the past or future black hole singularities, outside of $I_\mathrm{CFT's}$.

We will generalize this construction away from the thermofield double state to the case of arbitrary states and arbitrary regions $R$.  There might not be any singularities, but the basic point remains that the perturbations can be outside of $I_R$.  However, we can show that the perturbations cannot extend past $\mathcal{X}_{R}$:

\begin{thm}
The perturbation described in Eq. \ref{U's} cannot affect the bulk past $\mathcal{X}_{R}$, for any number $n$ of unitary operators.
\end{thm}
\begin{proof}
We will show this by induction.  In the last section we proved the case where $n = 1$.  We assume the statement holds for the $(n-1)$ rightmost signals in Eq. \ref{U's}, and let \nolinebreak $^{n-1}C_R$ be the causal surface in the resulting bulk spacetime.  We then consider the effects of the $n$-th signal, using the fact that $\mathcal{X}_{R}$ is deeper than \nolinebreak $^{n-1}C_R$ (Theorem \ref{causal}). 

Since the $n$-th signal is unitary, it does not affect the entropy $S_\mathrm{ent}$ on any side of any bulk surface anchored to $\partial R$ lying outside of $^{n-1}I_R$.  Therefore, the location of $\mathcal{X}_{R}$ cannot be affected through its dependence on $S_\mathrm{ent}$.  

$\mathcal{X}_{R}$ also depends on the area $A$.  But the $n$-th signal does not change the geometry outside of \nolinebreak $^{n-1}I_R$ either, by causality and the fact that $\mathcal{X}_{R}$ is outside of $I_R$.\footnote{There may be multiple quantum extremal surfaces, but if so, similar reasoning shows that they cannot be created or destroyed.}

Since $\mathcal{X}_{R}$ cannot be affected, \emph{a fortiori} the region behind it cannot be affected either.  Hence, assuming that the first $n-1$ unitary operators cannot modify the spacetime deeper than $\mathcal{X}_{R}$, neither can the $n$-th, proving the result.  To summarize, each signal can change the bulk up to $C_R$, and it can change the location of $C_R$, but it cannot change the fact that $C_R$ is closer to the boundary than $\mathcal{X}_{R}$
\end{proof}

Therefore $\mathcal{X}_{R}$, its entropy, its geometry, and the geometry behind it cannot be affected by any unitary signals propagating in from $R$.  This defines a sacrosanct region which cannot be causally influenced by $R$.  For example, in the case of an eternal black hole with an Einstein-Rosen bridge, even if we allow signals to be sent in from both sides, the geometry of the extremal surface lodged in its throat cannot be changed by the Shenker-Stanford construction.  Nor can any new extremal surfaces be created.  If we start with a throat with two extremal surfaces in the throat (e.g. because there is a bag of gold \cite{Wheeler,  Hsu:2008yi} in the interior), the region in between them is sacrosanct.  This provides some evidence that the information in the bag of gold is not in fact contained in the CFT, as suggested in \cite{Freivogel:2005qh, Marolf:2008tx, Marolf:2012xe}, although a proof of this fact would require also analyzing operators which do not propagate into the bulk causally.

\section{Barriers to Quantum Extremal Surfaces}\label{Barriers}
%
In a recent paper \cite{EngelhardtWall}, we found that classical geometries feature extremal surface barriers, i.e. regions which limit the reach of extremal surfaces anchored within a boundary region $R$.  In some cases, there are even barriers when $R$ is the entire boundary, preventing any boundary-anchored extremal surface from accessing an entire portion of the bulk.  In this section we will extend several of our results to quantum extremal surfaces.

Since the established AdS/CFT dictionary translates extremal surfaces into non-local boundary observables, the blindness of such probes is especially troubling in the case where the reach of local observables is likewise limited.  As in the previous section, this situation raises the question of whether the dual field theory contains any information about the blind region.


We shall use terminology from \cite{EngelhardtWall} to qualify quantum barriers. Define a splitting surface $M$ to be a codimension 1 surface which divides the spacetime into two regions, Ext$(M)$ and Int$(M)$. Given a splitting surface $M$, a quantum extremal surface $\mathcal{X}_{i}$ is said to be $M$-deformable if it can be continuously deformed from an initial surface $\mathcal{X}_{0}$ which lies entirely in Ext$(M)$, via a family of surfaces $\{\mathcal{X}_{i}\}$ each anchored to the boundary in Ext$(M)$.

Using Theorem \ref{Wall12Thm1}, we  prove the following results about the reach of quantum extremal surfaces:
\begin{thm} \label{NegBarrier}Let $M$ be a smooth, null, splitting surface.  If every spacelike slice $Q$ of $M$ is quantum trapped, i.e. for any point $q$ on $Q$,
\begin{equation}
\left. \frac{\delta S_\mathrm{gen}(M)}{\delta Q^{a}}k^{a}\right|_{q}<0 
\end{equation}
then no $M$-deformable quantum extremal surface ever touches $M$.  
\end{thm}
\begin{proof} Let $\{\mathcal{X}_{i}\}$ be a family of $M$-deformable surfaces as described above, and assume that there exists a surface $\mathcal{X}_{I}\in\{\mathcal{X}_{i}\}$ such that $\mathcal{X}_{I}$ crosses $M$, i.e. $\mathcal{X}_{I} \cap \text{Int}\left ( M\right)\neq \emptyset$. Since $\mathcal{X}_{0}$ and $\mathcal{X}_{I}$ are related via a series of smooth deformations, there is a ``midway'' surface $\mathcal{X}$ which coincides with $Q$ at some set of points $\{q_{\alpha}\}$ and is tangent to $Q$ at those points. Let $q\in \{q_{\alpha}\}$ be one such point.  Let $N$ be the codimension 1 null surface generated by null congruences shot from $\mathcal{X}$. In a small neighborhood of $q$, the infinitesimal normals to $Q$ and $\mathcal{X}$ agree, and the null generators on $N$ and $M$ likewise agree. By Theorem \ref{Wall12Thm1}, there exists a variation $\delta Q^{a}$ along the shared normal directions, such that at some point $p$ in a neighborhood of $q$:
\begin{equation}\left .\frac{\delta S_\mathrm{gen}\left(M\right)}{\delta Q^{a}}k^{a} \right | _{p}\geq \left . \frac{\delta  S_\mathrm{gen}\left(N\right)}{\delta Q^{a}}k^{a}\right | _{p} =\left . \frac{\delta S_\mathrm{gen}\left(N\right)}{\delta \mathcal{X}^{a}}k^{a} \right | _{p}\ =0  \end{equation}
where the last equality follows from the definition of a quantum extremal surface. This contradicts the assumption that at any point on $Q$, the generalized entropy is decreasing. Therefore $M$ cannot coincide with $\mathcal{X}$ at any point. 
%
\end{proof}

An analogous theorem can be proven in the non-generic case where the generalized entropy is unchanging on a null surface (e.g. a stationary horizon in the Hartle-Hawking state):

\begin{thm} \label{ZeroBarrier}Let $M$ again be a smooth null splitting surface.  If every spacelike slice $Q$ of $M$ is quantum {\bf marginally} trapped, i.e. for any point $q$ on $Q$,
\begin{equation} \frac{\delta S_\mathrm{gen}}{\delta Q^{a}}k^{a}=0\end{equation}
\noindent then no $M$-deformable connected quantum extremal surface can ever cross or touch $M$.\end{thm}
\begin{proof} 
Let $\{\mathcal{X}_{i}\}$ be again a family of $M$-deformable surfaces anchored to Ext$(M)$, with $X_{0}\subset\text{Ext}(M)$ and $\mathcal{X}_{I}\cap\text{Int}(M)\neq \emptyset$. Take $\mathcal{X}$ as before to be the ``midway'' surface, and let $N$ be the codimension 1 null surface generated by null congruences shot from $\mathcal{X}$. Pick $Q$ to be a spatial slice of $M$ 
which passes through each point $p$ where $\mathcal{X}$ touches $M$.  

At a coincident point $p$:
\begin{equation} \left .\frac{\delta S_\mathrm{gen}\left(M\right)}{\delta Q^{a}}\right | _{p} \geq \left .\frac{\delta S_\mathrm{gen}\left(N\right)}{\delta Q^{a}}\right | _{p}= \left .\frac{\delta S_\mathrm{gen}\left(N\right)}{\delta \mathcal{X}^{a}}\right | _{p}=0\label{Equality}\end{equation}
\noindent By Theorem \ref{Wall12Thm1}, Eq. \ref{Equality} can only be saturated if $N$ and $M$ coincide on an entire neighborhood of $p$. We can repeat the analysis above at the edge of each neighborhood, so that $M$ and $N$ (and therefore $\mathcal{X}$ and $Q$) must coincide everywhere. Since $M$ and $N$ are smooth, codimension 1 null surfaces, we find that $N=M_{c}$, where $M_{c}$ is the connected component of $M$ containing $p$. However, this implies that $\mathcal{X}$ lies entirely on $M$, and therefore ${\cal X}$ cannot be anchored to Ext$(M)$, which is a contradiction.  We arrive at the conclusion that no surface in $\{\mathcal{X}_{i}\}$ can ever cross or touch $M$.
\end{proof}

The proof of Theorem \ref{ZeroBarrier} gives rise to an immediate corollary: 

\begin{cor} \label{ZeroCor} Let $M$ be a null splitting surface whose slices are all quantum marginally trapped, as above.  If ${\cal X}$ is a connected quantum extremal surface that touches and is tangent to $M$ at a point, then ${\cal X}$ lies on $M$.\end{cor}
\begin{proof} The proof follows immediately from the proof of Theorem \ref{ZeroBarrier}. \end{proof}

Theorems \ref{NegBarrier} and \ref{ZeroBarrier} establish that there can be regions of spacetime which are inaccessible to quantum extremal surfaces anchored at some region $R$ on the boundary. In particular, when $R$ is the entire boundary, a quantum barrier defines a region of the bulk spacetime which is entirely inaccessible to any boundary-anchored quantum extremal surface, and---if our holographic entanglement proposal in section \ref{EE} is true---also cannot be probed using the entanglement entropy of the boundary region.  

In the classical case, a (codimension 2) extremal surface can itself be used to construct a barrier for other extremal surfaces, under the assumption of the null energy condition \cite{EngelhardtWall}.  This result can also be extended to quantum extremal surfaces.  Let $N^+$ be a null surface shot out from some ${\cal X}$ in the future-outward direction, and let $N^-$ be a null-surface shot out in the past-outward direction.  Their union $N = N^+ \cup N^-$ is a barrier surface.

Initially, as one moves along either of $N^\pm$, the generalized entropy is unchanging.  It will be conjectured in a future paper \cite{EntropicFocusing} that---at least for free fields---if the generalized entropy of a null surface is nonincreasing, it continues to be nonincreasing thereafter (and if it begins to decrease, it continues to do so).  This helps to explain the result (proven in \cite{Wall12} from the GSL) that null surfaces generated by quantum trapped surfaces must inevitably terminate at finite values of the affine parameter.  It is analogous to the classical case, where the area must focus because of the Raychaudhuri equation.

Assuming that this quantum focusing property applies, we see that $N$ has the requisite behavior for a barrier:
\begin{equation} \frac{\delta S_{\text{gen}}(N)}{\delta Q^{a}(p)} k^a \leq 0,\end{equation}
$k^a$ being an outwards pointing null vector.  If there were a family of quantum extremal surfaces which could be deformed across $N$, it would have to cross either $N^+$ or $N^-$ first.  Whichever is the case, by Theorems \ref{NegBarrier} and \ref{ZeroBarrier}, we get a contradiction.  So $N$ is a barrier to $N$-deformable quantum extremal surfaces anchored in Ext$(N)$.


%

\section{Discussion: Limits on Bulk Reconstruction}\label{discussion}

We have proven a number of results for quantum extremal surfaces, and proposed that their generalized entropy is likely to be the correct extension of the HRT and FLM proposals to quantum bulk spacetimes at subleading order in $\hbar$.  This corresponds to calculating corrections to the boundary entanglement entropy which are subleading in $N$. It may be possible to prove additional theorems about quantum extremal surfaces using the maximin approach \cite{WallMaximin}.

We will now discuss how these results fit into the bigger picture of bulk reconstruction.  By analogy to $D_R$ and $I_R$, we will define $\mathbf{D}_R$ and $\mathbf{I}_R$ as the regions which can be reconstructed/influenced by $R$ using the AdS/CFT dictionary.

An important open question in AdS/CFT is how much of the bulk can be reconstructed from a specified boundary region $R$ \cite{HubenyRangamani, Bousso:2012sj, Czech:2012bh}.  It seems that this reconstructable region ${\bf D}_R$ should be at least as large as the causal wedge $W_R$, and plausibly should extend all the way to the extremal surface $X_R$ \cite{Wall12}, at least for classical bulk geometries.  In the case of quantum spacetimes, it is natural to generalize this idea to the quantum extremal surface ${\cal X_R}$.

A related question is how large of a bulk region can be \emph{influenced} by operations in the region $R$.  This region of influence ${\bf I}_R$ must be at least as large as the domain of influence $I_R$, due to the ability to send signals into the bulk causally, but the Shenker--Stanford construction \cite{ShenkerStanford} shows that it must be larger than that.  

In fact---assuming that there are no operators in $D_R$ which commute with everything else in $D_R$---this region of bulk influence ${\bf I}_R$ must be at least as large as ${\bf D}_R$, since one can freely act on $D_R$ with unitary operators in order to change it into anything one likes.  However, ${\bf I}_R$ cannot extend into the complementary domain of reconstruction 
${\bf D}_{\bar{R}}$, or else unitaries in $R$ would be able to affect the spacelike separated region $\bar{R}$.  

For a classical pure bulk $X_R = X_{\bar{R}}$, so \emph{if} the reconstructable region extends up to $X_R$, then ${\bf D}_R$ and ${\bf D}_{\bar{R}}$ are also complementary.  That would imply that ${\bf I}_R = I^+({\bf D}_R) \cup I^-({\bf D}_R)$, so that the two concepts would be essentially the same, except that ${\bf I}_R$ would also include the region timelike to $X_R$.  But it is unclear whether the needed assumptions are in fact true.  These same reflections would also apply to the quantum extremal surface ${\cal X}$.

Our results provide evidence that quantum extremal surfaces ${\cal X}$ are closely related to the regions of bulk reconstruction and influence.  In section \ref{reach} we used the GSL to prove that causally propagating signals cannot go farther into the bulk than the closest quantum extremal surface ${\cal X}_R$.  To the extent that this propagating signal is a good measure of what unitary operators can do, this suggests that the domain of influence extends no farther than the quantum extremal surface ${\cal X}$.

Another probe of the spacetime is the entanglement entropy itself, which is a good example of a nonlocal bulk observable.  In section \ref{Barriers} we showed that in many spacetimes, there exist barrier surfaces which cannot be crossed by any other quantum extremal surface ${\cal X}_i$ anchored within some boundary region $D_R$, at least when ${\cal X}_i$ is continuously deformable to the region outside ${\cal X}_R$.

It is therefore natural to suppose that the reconstructable region ${\bf D}_R$ extends up to the outermost barrier itself.  Using some results from a future article \cite{EntropicFocusing}, we also argued that a null surface $N$ shot out from a quantum extremal surface ${\cal X}_R$ has the requisite properties to be a barrier to other quantum extremal surfaces.
 
Thus there are several indicators that ${\cal X}_R$ forms some type of boundary surface relevant to questions of bulk reconstruction and influence from $R$.  We should emphasize that we have only explored these questions using a limited subset of probes, so it is possible that new regions of the bulk would be accessible to a different kind of probe.

In the case where $R$ is taken to be the entire boundary, these results suggest that it may not be possible to reconstruct the region behind the extremal surface of the entire boundary ${\cal X}_\mathrm{CFT}$ from the dual CFT.  If that is really true, then it suggests that there might be so-called "superselection" information in the bulk which is not contained in the boundary CFT \cite{Freivogel:2005qh, Marolf:2008tx, Marolf:2012xe}.  Alternatively, there might be a firewall at ${\cal X}_\mathrm{CFT}$ which destroys in-falling observers and obviates the need for a region behind it \cite{AMPS, EngelhardtWall}.

Our results depend heavily on the fact that we are only interested in proving inequalities, since this allowed us to use Theorem \ref{Wall12Thm1} to neglect higher-curvature effects in the entropy functional (cf. section \ref{Premises}) as compared to the Bekenstein-Hawking area term.  If the higher-curvature terms were instead of the same order as the area term, many of our results would become much more difficult to prove; perhaps even false.

There is a hypothesis that in full quantum gravity, the entropy of a stationary black hole is in fact entirely due to entanglement entropy outside the horizon \cite{SusskindUglum,Frolov:1993ym, Barvinsky:1994jca,  Frolov:1998vs}.  That would mean that $S_\mathrm{gen}$ would be the fine-grained entropy of the exterior quantum gravity microstates.  If we apply this hypothesis to quantum extremal surfaces, we would conclude that ${\cal X}$ is really obtained by extremizing the entanglement entropy.  This would \emph{formally} allow the various proofs in this article to be applied using only the information-theoretical parts associated with $S_\mathrm{ent}$.  But this argument places us squarely in the non-perturbative quantum gravity regime, making usual metric concepts highly suspect.

Besides black hole thermodynamics, one of the few clues we have about nonperturbative quantum gravity is AdS/CFT.  The entanglement entropy $S(R)$ of a CFT region is just as well-defined at small $N$ and weak coupling as at large $N$ and strong coupling.  If AdS/CFT is valid in the former regime, there must be some concept in the bulk quantum gravity theory which is dual to $S(R)$.  It is interesting to ask whether the concept of a quantum extremal surface can be generalized to that context.  Perhaps progress can be made by assuming that some suitably quantum notion of holographic entanglement surfaces survives, and using that to explore the meaning of  \nolinebreak quantum \nolinebreak geo\nolinebreak\phantom{}metry.

\vskip 2cm
\centerline{\bf Acknowledgements}
\vskip 1cm
\noindent It is a pleasure to thank Don Marolf, Ahmed Almheiri, Gary Horowitz, Raphael Bousso, Zachary Fisher, William Kelly, Dalit Engelhardt, Mark Van Raamsdonk, and Mudassir Moosa for helpful discussions. This work is supported in part by the National Science Foundation under Grant No. PHY12-05500. NE is also supported by the National Science Foundation Graduate Research Fellowship under Grant No. DGE-1144085. AW is also supported by the Simons Foundation.


\end{document}